\documentclass[11pt, a4paper]{article}

\usepackage[a4paper]{geometry}
\usepackage{amsmath,amssymb, amsthm, bbm, mathtools, wasysym, graphicx}

\usepackage[utf8]{inputenc}

\newtheorem{thm}{Theorem}[section]

\newtheorem{lem}[thm]{Lemma}
\newtheorem{prop}[thm]{Proposition}
\newtheorem{remark}[thm]{Remark}

\numberwithin{equation}{section}

\newenvironment{enum(a)}{\begin{enumerate}}{\end{enumerate}}
\newenvironment{enum(i)}{\begin{enumerate}}{\end{enumerate}}

\newcommand{\CCC}{{\mathbb{C}}}

\newcommand{\NNN}{{\mathbb{N}}}

\newcommand{\RRR}{{\mathbb{R}}}

\newcommand{\ZZZ}{{\mathbb{Z}}}

\newcommand{\cB}{{\mathcal{B}}}

\newcommand{\cF}{{\mathcal{F}}}

\newcommand{\cO}{{\mathcal{O}}}

\newcommand{\cS}{{\mathcal{S}}}
\newcommand{\cT}{{\mathcal{T}}}

\begin{document}

\author{Alexisz Tamás Gaál}
\title{Long range order in a hard disk model\\ in statistical mechanics}
\date{October 2, 2013}

\maketitle

\small
\begin{abstract}
\noindent We model two-dimensional crystals by a configuration space in which every admissible configuration is a hard disk configuration and a perturbed version of some triangular lattice with side length one. In this model we show that, under the uniform distribution, expected configurations in a given box are arbitrarily close to some triangular lattice whenever the particle density is chosen sufficiently high. This choice can be made independent of the box size.
\end{abstract}


\normalsize

\section{Introduction}

The breaking of rotational symmetry in two-dimensional models of crystals at low temperature has been indicated since long, see \cite{Mer68} and \cite{NH79}. F. Merkl and S. W.W. Rolles showed the breaking of rotation symmetry in \cite{MR} in a simple model without defects. In this model of crystals, atoms can be enumerated by a triangular lattice. In the very recent work \cite{HMR} by M. Heydenreich, F. Merkl and S. W.W. Rolles, defects were integrated into the model; defects are single, isolated, missing atoms. However, the results in \cite{HMR} can be generalized to larger bounded islands of missing atoms as also mentioned in \cite{HMR}, but non-local defects are not included. The first model in \cite{MR} treated pair potentials with at least quadratic growth; the second one, \cite{HMR}, tackled the case of strictly convex potentials.

We are going to examine an analogue of the models in \cite{MR} and \cite{HMR} with a hard-core repulsion. For this potential we show the breaking of the rotational symmetry in a strong sense. Our model does not include defects, but the result extends to models with isolated defects as in \cite{HMR}. Uniformity in the box size ensures the existence of infinite volume measures with the analogous property. This work is motivated by the following open problem: is there a Gibbs measure on the set of locally finite point configurations in $\RRR^2$ which breaks the rotational symmetry of the hard-core potential? This question is analogous to the problem which was solved in \cite{Geo99} and \cite{Ri09} for translational symmetry. However, the outcome is different than what is expected in the case of rotational symmetry, as translational symmetry is preserved, see \cite{Geo99} and \cite{Ri09}.

\section{Configuration space}

The \emph{standard triangular lattice} in $\RRR^2$ is the set $I=\ZZZ+\tau \ZZZ$ with $\tau=e^{\frac{i\pi}{3}}$. We identify $\ZZZ\subset\RRR\subset\RRR^2$ by $\RRR\ni x\ \hat=\ (x,0)\in\RRR^2$ and $\RRR^2\subset\CCC$ by $(x,y)\ \hat=\ x+iy$. The set $I$ is an index set, which is going to be used to parametrize countable point configurations in the real plane. Let us define the quotient space $I_N=I/(NI)$ for an $N\in\NNN:=\{1,2,3, ...\}$. We identify $I_N$ with the following specific set of representatives: 

\begin{equation}
I_N=\{x+y\tau\ |\ x,y\in\{0, ..., N-1\} \}.
\label{I_N}
\end{equation}

A \emph{parametrized point configuration} in $\RRR^2$ is a function $\omega: I\rightarrow \RRR^2$, $x\mapsto \omega(x)$, which determines the point configuration $\{ \omega(x)\ |\ x\in I \}\subset \RRR^2$.  For the set of all parametrized point configurations we introduce the character $\Omega=\{\omega: I\rightarrow \RRR^2\}$. Note that a single point configuration $\{ \omega(x)\ |\ x\in I \}\subset \RRR^2$ can be parametrized by many different $\omega\in\Omega$.

Let $\epsilon\in(0,1]$. An \emph{$N$-periodic parametrized point configuration} with side length $l\in(1, 1+\epsilon)$ is a parametrized configuration $\omega$ which satisfies the \emph{periodic boundary conditions}:

\begin{equation}
\omega(x+Ny)=\omega(x)+lNy \quad \textrm{for all}\ x, y \in I.
\label{periodic}
\end{equation}

The set of $N$-periodic parametrized configurations with side length $l$ is denoted by $\Omega^{per}_{N, l}\subset \Omega$. From now on we will omit the word parametrized because we are going to work solely with \emph{point configurations} which are parametrized by $I$. An $N$-periodic configuration is uniquely determined by its values on $I_N$. Therefore, we identify $N$-periodic configurations $\omega\in\Omega^{per}_{N,l}$ with functions $\omega: I_N\rightarrow \RRR^2$.

The bond set $E\subset I\times I$ contains index-pairs with Euclidean distance one; this is $E=\{(x,y)\in I\times I\ |\ |x-y|=1\}$. In order to transfer the definition to the quotient space $I_N$, we define an equivalence relation $\sim_N$ on $E$ by $(x, y)\ \sim_N \ (x',  y' )$ if and only if there is a $z\in N I$ such that $x=x'+z$ and $y=y'+z$. We set $E_N=E/\sim_N$. We can think of $E_N$ as a bond set $E_N\subset I_N\times I_N$.

For $x\in I$ and $z\in \{1, \tau\}$, define the open \emph{triangle}

\[
\triangle_{x,z}=\{x+sz+t\tau z\ |\ 0<s,t,\ s+t<1\}
\]

with corner points $x, \ x+z$ and $x+\tau z$. For $\triangle_{x,z}$ denote the set of corner points by $\cS(\triangle_{x,z})=\{x, x+z, x+\tau z\}$. On the set of all triangles

\[
\cT=\{\triangle_{x,z} \ |\ x\in I\ \textrm{and}\ z\in\{1, \tau\}\},
\]

we define an equivalence relation: $\triangle_{x,z}\sim_N \triangle_{x',z'}$ if and only if $x-x'\in NI$ and $z=z'$. The set of equivalence classes is denoted by $\cT_N=\cT/\sim_N$. We identify equivalence classes $\triangle\in\cT_N$ with their unique representative with corners in the set $\{x+\tau y\ |\ x,y\in\{0, ..., N\}\}$. The closures of the triangles in $\cT_N$ cover the convex hull of the above set, which is denoted by $U_N=\textrm{conv}(\{x+\tau y\ |\ x,y\in\{0, ..., N\}\})$.

\section{Probability space}

By definition $\Omega=(\RRR^2)^I$, and we can identify $\Omega^{per}_{N, l}=(\RRR^2)^{I_N}$. Both sets are endowed with the corresponding product $\sigma$-fields $\cF=\bigotimes_{x\in I}\cB(\RRR^2)$ and $\cF_N=\bigotimes_{x\in I_N}\cB(\RRR^2)$ where $\cB(\RRR^2)$ denotes the Borel $\sigma$-field on each factor.
The event of admissible, N-periodic configurations $\Omega_{N, l}\subset \Omega^{per}_{N, l}$ is defined by the properties $(\Omega 1)-(\Omega 3)$:

$(\Omega 1) \quad |\omega(x)-\omega(y)|\in (1, 1+\epsilon)$ for all $(x,y)\in E$.

For $\omega\in\Omega$ we define the extension $\hat\omega: \RRR^2\to\RRR^2$ such that $\hat\omega(x)=\omega(x)$ if $x\in I$, and on the closure of any triangle $\triangle\in\cT$, the map $\hat\omega$ is defined to be the unique affine linear extension of the mapping defined on the corners of $\triangle$.

$(\Omega 2)\quad$ The map $\hat\omega: \RRR^2\to\RRR^2$ is injective.

$(\Omega 3) \quad$ The map $\hat\omega$ is orientation preserving, this is to say that $\det(\nabla\hat\omega(x))>0$ for all $\triangle\in\cT$ and $x\in\triangle$ with the Jacobian $\nabla\hat\omega: \cup\cT\to\RRR^{2\times2}$.

Define the set of \emph{admissible, $N$-periodic configurations} as

\[
\Omega_{N, l}=\{\omega\in\Omega_{N, l}^{per}\ |\ \omega\ \textrm{satisfies}\ (\Omega1)\textrm{--}(\Omega3)\}
\]

and the set of all \emph{admissible configurations} as $\Omega_\infty=\{\omega\in\Omega\ |\ \omega\ \textrm{satisfies}\ (\Omega1)\textrm{--}(\Omega3)\}$. Note that for $\omega\in\Omega_{N,l}^{per}$, $(\Omega 2)$ is fulfilled if and only if  $\hat\omega$ is a bijection. This observation is a consequence of the periodic boundary conditions (\ref{periodic}) and the continuity of $\hat\omega$.

\begin{figure}[!ht]
  \centering
    \includegraphics[width=0.8\textwidth]{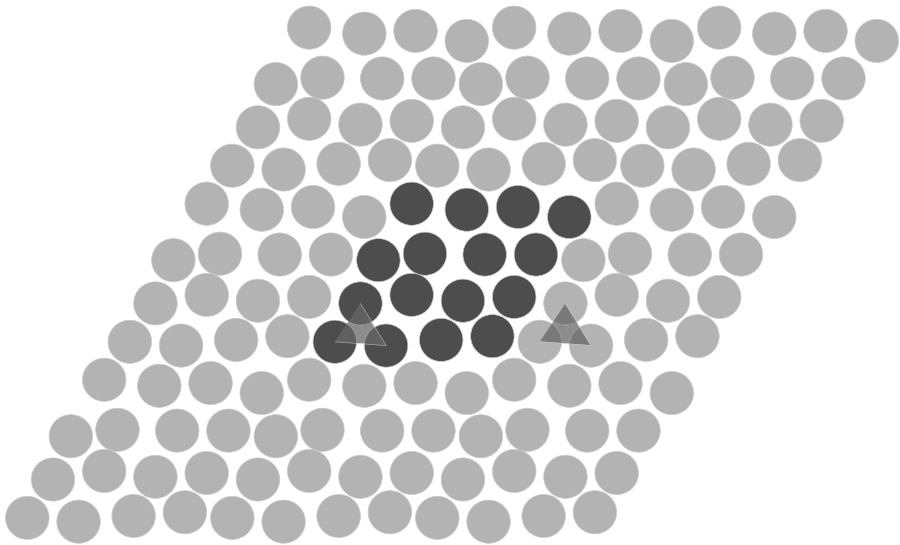}
    \caption{A part of an admissible, $4$-periodic configuration.}
    \label{figure}
\end{figure}

The set $\Omega_{N, l}$ is non-empty and open in $(\RRR^2)^{I_N}$. The scaled standard configuration $\omega_l(x)=lx$, for $x\in I$ and $1<l<1+\epsilon$, is an element both of $\Omega_{N, l}$ and $\Omega_{\infty}$. Figure \ref{figure} illustrates a part of an admissible, $4$-periodic configuration. The points of the configuration are illustrated by hard disks with radii 1/2. The image of $I_4$ and those of two equivalent triangles are shaded in the figure.

Clearly,  $0<\delta_0\otimes\lambda^{I_N \setminus\{0\}} (\Omega_{N, l})<\infty$ with the Lebesgue measure $\lambda$ on $\RRR^2$ and the Dirac measure $\delta_0$ in $0\in\RRR^2$. The lower bound holds because sections of $\Omega_{N, l}$ are non-empty and open in $(\RRR^2)^{I_N\setminus\{0\}}$ if $\omega(0)$ is fixed; the upper bound is a consequence of the parameter $\epsilon$ in $(\Omega 1)$. Let the probability measure $P_{N, l}$ be

\[
P_{N, l}(A)=\frac{\delta_0\otimes\lambda^{I_N \setminus\{0\}} (\Omega_{N, l}\cap A)}{\delta_0\otimes\lambda^{I_N \setminus\{0\}} (\Omega_{N, l})}
\]

for any Borel measurable set $A\in \cF_N$, thus $P_{N, l}$ is the uniform distribution on the set $\Omega_{N, l}$ with respect to the \emph{reference measure} $\delta_0\otimes\lambda^{I_N \setminus\{0\}}$. The first factor in this product refers to the component $\omega(0)$ of $\omega\in\Omega$. We call the measures $P_{N, l}$ \emph{finite-volume Gibbs measures} and the parameter $l$ in the definition of $\Omega_{N,l}$ and $P_{N,l}$ is the \emph{pressure parameter} of the system. In fact, the pressure parameter $l$ controls the density of periodic configurations, and therefore is inversely related to the physical pressure of the system.

\section{Result}

We have the following finite-volume result:

\begin{thm} \label{thm1} For $\epsilon$ sufficiently small $($such that equation \textnormal{($\ref{area1}$)} holds for all $1<a_i<1+\epsilon)$, one has

\begin{equation}
\lim_{l\downarrow 1}\sup_{N\in\NNN}\sup_{\triangle\in\cT_N} E_{P_{N,l}}[\ |\nabla\hat\omega(\triangle)-\textnormal{Id} |^2\ ]=0
\label{fnvolume}
\end{equation}

with the constant value of the Jacobian $\nabla\hat\omega(\triangle)$ on the set $\triangle\in\cT_N$.
\end{thm}

Weak limits of $(P_{N,l})_{N\in\NNN}$ are called \emph{infinite-volume Gibbs measures}. Since the convergence in Theorem \ref{thm1} is uniform in $N$, there is an infinite-volume Gibbs measure $P$ such that $E_{P} [ \ |\nabla\hat\omega(\triangle)-\textrm{Id} |^2\ ]$ is small on every triangle $\triangle\in\cT$. This is actually a result about a spontaneous breaking of the rotational symmetry in a strong sense. The set $\Omega_\infty$ is rotational-invariant, and this symmetry is broken by some infinite-volume Gibbs measure as per (\ref{fnvolume}). Spontaneous breaking of the rotational symmetry in the usual sense can be proved immediately. This observation is formulated and proved in the next proposition. A similar result and its proof is also mentioned in \cite[Section 1.3]{HMR}.

\begin{prop}
For all $l\in(1, 1+\epsilon)$, $N\in \NNN$, $x\in I$ and $z\in I$ with $(0,z)\in E$, we have

\begin{align}
E_{P_{N,l}}[\omega(x+z)-\omega(x)]=lz.
\label{easysymmbreak}
\end{align}

\begin{proof}
We follow the ideas stated in \cite[Section 1.3]{HMR}. The reference measure $\delta_0 \otimes \lambda^{I_N\setminus \{0\}}$ is invariant under the bijective translations

\begin{equation}
\psi_b: \Omega^{per}_{N, l}\rightarrow \Omega^{per}_{N,l} \quad (\omega(x))_{x\in I}\mapsto (\omega(x+b)-\omega(b))_{x\in I}
\label{translation}
\end{equation}

for all $b\in I$. The set $\Omega_{N,l}$ is also invariant under $\psi^{-1}_b=\psi_{-b}$. As a consequence, the measures $P_{N,l}$ are invariant under $\psi_b$ for all $b\in I$, and the random vectors $\omega(x+z)-\omega(x)$ have the same distribution under $P_{N,l}$ for all $x\in I$ and a fixed $z$. Therefore, we obtain (\ref{easysymmbreak}) from the periodic boundary conditions (\ref{periodic}).
\end{proof}
\end{prop}

The expression $|\omega(x+z)-\omega(x)|$ is $P_{N,l}$-almost surely uniformly bounded in $N$, hence (\ref{easysymmbreak}) carries over to weak limits of $P_{N,l}$ as $N\to\infty$. Consequently, such weak limits are not rotational-invariant. However, in the next section, we show Theorem \ref{thm1}, which states symmetry breaking in a much stronger sense.

\section{Proof}

As in \cite{HMR}, the central argument is the following rigidity theorem from \cite[Theorem 3.1]{Fries}, which generalizes Liouville's Theorem.

\begin{thm}[Friesecke, James and M\"uller]
Let $U$ be a bounded Lipschitz domain in $\RRR^n, \ n\geq 2$. There exists a constant $C(U)$ with the following property: For each $v\in W^{1,2}(U, \RRR^n)$ there is an associated rotation $R\in \textnormal{SO($n$)}$ such that

\begin{equation*}
||\nabla v-R||_{L^2(U)}\leq C(U)||\textnormal{dist}(\nabla v, \textnormal{SO}(n))||_{L^2(U)}.
\end{equation*}

\label{fries}
\end{thm}

Liouville's Theorem states that a function $v$, fulfilling $\nabla v(x)\in\mathrm{SO}(n)$ almost everywhere, is a rigid motion. Theorem \ref{fries} generalizes this result. We are going to set $v=\hat\omega|_{U_N}$ and $U=U_N$, which is a bounded Lipschitz domain. The function $\hat\omega|_{U_N}$  is affine linear on each triangle $\triangle\in\cT_N$, thus piecewise affine linear on $U_N$. As a consequence, $\hat\omega|_{U_N}$ belongs to the class $W^{1,2}(U_N, \RRR^n)$. The following remark, which also appears in \cite{Fries} at the end of Section 3, is essential to achieve uniformity in Theorem \ref{thm1} in the parameter $N$.

\begin{remark}

The constant $C(U)$ in Theorem \ref{fries} is invariant under scaling of the domain: $C(\alpha U)=C(U)$ for all $\alpha>0$. By setting $v_\alpha(\alpha x)=\alpha v(x)$ for $x\in U$, we have $\nabla v_\alpha(\alpha x)=\nabla v(x)$, and therefore $||\nabla v_\alpha-R||_{L^2(\alpha U)}=\alpha^{n/2}||\nabla v-R||_{L^2(U)}$, and $||\textnormal{dist}(\nabla v_\alpha, \textnormal{SO}(n))||_{L^2(\alpha U)}\ =\ \alpha^{n/2}\ ||\textnormal{dist}(\nabla v, \textnormal{SO}(n))||_{L^2(U)}$. Consequently, the constants $C(U_N)$ for the domains $U_N$ $(N\geq 1)$ can be chosen independently of $N$.
\label{remark}
\end{remark}

We are going to show that for $\omega\in\Omega_{N,l}$, the $L^2$-distance on $U_N$ of the Jacobian matrix $\nabla\hat\omega$ from the scaled identity matrix $l\ \textrm{Id}$ can be controlled by the difference of the areas of $\hat\omega(U_N)$ and $U_N$. Due to the periodic boundary conditions, $\lambda(\hat\omega(U_N))$ does not depend on configurations $\omega$ with $(\Omega 2)$, thus the mentioned area difference provides a suitable uniform control on the set $\Omega_{N,l}$. First, we show that the $L^2$-distance of $\nabla\hat\omega$ from the scaled identity $l\ \textrm{Id}$ can be controlled by the sum over the squared deviations of the triangles' side lengths from one. The one should be associated with the side length of an equilateral triangle. To achieve this estimate, we will apply the rigidity theorem, Theorem \ref{fries}, but first we cite an analogous result which holds locally on each triangle.

The following lemma provides the desired estimate on each triangle. It states that the distance from $\textrm{SO}(2)$ of a linear map near $\textrm{SO}(2)$ can be controlled by terms which measure how the linear map deforms the side lengths of a standard equilateral triangle.

\begin{lem}
There is a positive constant $C$ such that, for all linear maps  $A: \RRR^2 \rightarrow \RRR^2$ with $\textnormal{det}(A)>0$ and the property

\begin{equation}
||Av_i|-1|\leq 1 \quad \textrm{for all} \ i\in\{1,2,3\} 
\label{requirement}
\end{equation}

where $v_1=(1, 0)$, $v_2=(\frac{1}{2}, \frac{\sqrt3}{2})$, $v_3=v_1-v_2$, the following inequality holds:

\begin{equation}
\textnormal{dist}\left(A\ ,\ \textnormal{SO}(2)\right)^2:=\inf_{R\in \textnormal{SO}(2)} \left|A-R\right|^2\leq C \max_{i\in\{1, 2, 3\}} ||Av_i|-1|^2
\label{triangleeq}
\end{equation}

where $|M|=\sqrt{\textnormal{tr}(M^tM)}$ is the Frobenius norm and $|v|$ is the Euclidean norm of $v$.
\label{triangle}
\end{lem}

A proof can be found in \cite[Lemma 4.2. in the appendix]{Theil}. In this proof the requirement (\ref{requirement}) is formulated by means of a positive constant $\alpha_0$: $||Av_i|-1|\leq \alpha_0 \quad \textrm{for all} \ i\in\{1,2,3\}$, although the proof also applies to the special case $\alpha_0=1$ as stated in Lemma \ref{triangle}.

Now, we prove the mentioned estimate, which provides control over the $L^2$-distance of $\nabla\hat\omega$ from the scaled identity matrix in terms of the side length deviations.

\begin{lem}
There is a constant $c$ such that for all $N\geq 1$ and $1<l<1+\epsilon$ the inequality

\begin{equation}
||\ \nabla\hat\omega-l\ \textnormal{Id}\ ||^2_{L^2(U_N)}\leq c \sum_{(x,y)\in E_N} (|\omega(x)-\omega(y)|-1)^2
\label{sideeq}
\end{equation}

holds for all $\omega\in\Omega_{N, l}$, and hence

\begin{equation}
E_{P_{N,l}}[\ ||\ \nabla\hat\omega-l\ \textnormal{Id}\ ||^2_{L^2(U_N)}\ ]\leq c \sum_{(x,y)\in E_N} E_{P_{N,l}}[\ (|\omega(x)-\omega(y)|-1)^2\ ]
\label{lemmaeq2}
\end{equation}

where the $L^2$-norm is defined with respect to some scalar product on $\RRR^{2\times 2}$, and $|\cdot|$ denotes the Euclidean norm on $\RRR^2$.
\label{side}
\end{lem}

Note that the right side in equation (\ref{sideeq}) is strictly positive because of the boundary conditions (\ref{periodic}) and because $l>1$, whereas the left is zero for $\omega=\omega_l\in\Omega^{per}_{N,l}$. Since the measure $P_{N,l}$ is supported on the set $\Omega_{N, l}$, (\ref{lemmaeq2}) follows from (\ref{sideeq}). Also note that $c$ does not depend on $N$.

\begin{proof}

Let $\omega\in\Omega_{N, l}$. By Lemma \ref{triangle} we conclude that on every triangle $\triangle\in\cT_N$, we have 

\begin{equation*}
\textrm{dist}\left(\nabla\hat\omega(\triangle) ,\ \textrm{SO}(2)\right)^2 \leq C \max_{x\not=y\in {S(\triangle)}} (|\omega(x)-\omega(y)|-1)^2\leq\frac{C}{2} \sum_{x\not=y\in {S(\triangle)}} (|\omega(x)-\omega(y)|-1)^2
\end{equation*}

where we used the assumption $\epsilon\leq 1$ together with $(\Omega 1)$ and $(\Omega 3)$ to apply Lemma \ref{triangle}. The factor $1/2$ is a consequence of summing over all non-equal pairs $(x,y)$. Orthogonality of the functions which are non-zero on different triangles gives

\[
||\ \textrm{dist}(\nabla\hat\omega, \textrm{SO}(2))\ ||^2_{L^2(U_N)}\leq c_1 \sum_{(x,y)\in E_N}(|\omega(x)-\omega(y)|-1)^2
\]

with $c_{1}=C\ \lambda(\triangle_{0, 1})=C \sqrt{3}/4$ because we sum again over both pairs $(x,y)$ and $(y,x)$ on the right side. With application of Theorem \ref{fries} about geometric rigidity, we find an $R(\omega)\in\textrm{SO}(2)$ such that

\[
||\ \nabla\hat\omega-R(\omega)\ ||^2_{L^2(U_N)}\leq c_{2} \ ||\ \textrm{dist}( \nabla\hat\omega, \textrm{SO}(2))\ ||^2_{L^2(U_N)},
\]

with a constant $c_{2}$, which does not depend on $N$ by Remark \ref{remark}. Due to the periodic boundary conditions (\ref{periodic}), the function $\hat\omega-l\ \textrm{Id}$ is $N$-periodic, this is to say

\begin{equation}
\hat\omega(x+Ny)-l (x+Ny)=\hat\omega(x)-lx \quad \textrm{for all } x\in \RRR^2 \textrm{ and } \ y \in I.
\label{Nperiodic}
\end{equation}

Let $A\in\RRR^{2\times 2}$ be a constant matrix. Integrating the function $\langle\nabla\hat\omega-l \ \textrm{Id}, A  \rangle$ over the set $U_N$, the result equals zero since, by (\ref{Nperiodic}) and the fundamental theorem of calculus,

\[
\int_0^1\langle\nabla\hat\omega-l\ \textrm{Id}, A  \rangle(x+tN)\textnormal{d}t=0 \quad \textrm{for all } x\in \RRR^2
\]

where we used the embedding $\RRR\subset\RRR^2$. Consequently, we obtain the orthogonality property: $\nabla\hat\omega-l\ \textrm{Id}\perp_{L^2(U_N)} A$, for any constant matrix $A\in\RRR^{2\times 2}$ and thus

\[
||\ \nabla\hat\omega-l\ \textrm{Id}\ ||^2_{L^2(U_N)}+||\ l\ \textrm{Id}-R(\omega)\ ||^2_{L^2(U_N)}=||\ \nabla\hat\omega-R(\omega)\ ||^2_{L^2(U_N)}
\]

by Pythagoras. Since $||\ l\ \textrm{Id}-R(\omega)\ ||^2_{L^2(U_N)}\geq0$ and because $P_{N,l}$ is supported on the set $\Omega_{N,l}$, the lemma is established with $c=c_{1}c_{2}$.
\end{proof}

With Lemma \ref{side} we can now prove Theorem \ref{thm1}.

\begin{proof}[Proof of Theorem \ref{thm1}]

Heron's formula states that the area $\lambda(\triangle)$ of the triangle $\triangle$ with side lengths $a_1, a_2, a_3$ is given by

\begin{equation}
\lambda(\triangle)=\frac{1}{4}\sqrt{(a_1+a_2+a_3)(-a_1+a_2+a_3)(a_1-a_2+a_3)(a_1+a_2-a_3)}.
\label{heron}
\end{equation}

By first order Taylor approximation of (\ref{heron}) at the point $a_i=1$, $i\in\{1,2,3\}$ we obtain

\begin{equation*}
\lambda(\triangle)-\lambda({\triangle_{0,1}})=\frac{1}{2\sqrt{3}}\sum_{i=1}^3(a_i-1)+o\left(\sum_{i=1}^3 |a_i-1|\right)\quad \textrm{as}\ (a_1, a_2, a_3)\to(1,1,1).
\end{equation*}

Since the function $\lambda$ is smooth in a neighborhood of $(1,1,1)$, we could also express the remainder term as Big $\cO$ of the sum of the squares. In the following we only need the weaker estimate on the remainder. We choose $\epsilon$ so small that the inequality

\begin{equation}
\frac{1}{4\sqrt{3}}\sum_{i=1}^3(a_i-1) \leq \lambda(\triangle)-\lambda({\triangle_{0,1}})
\label{area1}
\end{equation}

is satisfied whenever $1<a_i<1+\epsilon$. Note that we have divided the constant by two preceding the sum. Let us fix such an $\epsilon$ and assume that $\Omega^{per}_{N,l}$ is defined by means of this $\epsilon.$ Using (\ref{area1}) we can also estimate the squared side length deviations:

\begin{equation}
\sum_{i=1}^3(a_i-1)^2 \leq 4\sqrt{3}\ \epsilon\ (\lambda(\triangle)-\lambda({\triangle_{0,1}})).
\label{area}
\end{equation}

By equation (\ref{sideeq}) from Lemma \ref{side} and ($\ref{area}$), we get an upper bound on $||\nabla\hat\omega-l\ \textnormal{Id}||_{L^2(U_N)}^2$ in terms of the area differences. By summing up the contributions ($\ref{area}$) of the triangles $\triangle\in\cT_N$, we conclude for all $\omega\in\Omega_{N,l}$ that

\begin{equation}
||\ \nabla\hat\omega-l\ \textnormal{Id}\ ||_{L^2(U_N)}^2 \leq 4\sqrt{3} \ \epsilon\ c \sum_{\triangle\in\cT_N}(\lambda(\hat\omega(\triangle))-\lambda({\triangle_{0,1}})).
\label{maxarea}
\end{equation}

As a consequence of $(\Omega 2)$ and the periodic boundary conditions ($\ref{periodic}$), the right hand side in ($\ref{maxarea}$) does not depend on $\omega\in\Omega_{N,l}$. Hence, with $\omega_l\in\Omega_{N,l}$ we can compute

\begin{equation}
\sum_{\triangle\in\cT_N}(\lambda(\hat\omega(\triangle))-\lambda({\triangle_{0,1}}))=\sum_{\triangle\in\cT_N}(\lambda(\hat\omega_l(\triangle))-\lambda({\triangle_{0,1}}))=|\cT_N|\ \lambda(\triangle_{0,1})(l^2-1).
\label{maxarea1}
\end{equation}

The combination of the equations (\ref{maxarea}) and (\ref{maxarea1}) gives

\begin{equation}
||\ \nabla\hat\omega-l\ \textnormal{Id}\ ||_{L^2(U_N)}^2 \leq 4\sqrt{3} \ \epsilon\ c\ |\cT_N|\ \lambda(\triangle_{0,1})(l^2-1).
\label{maxarea2}
\end{equation}

The reference measure $\delta_0\otimes\lambda^{I_N\setminus\{0\}}$ and the set of allowed configurations $\Omega_{N,l}$ are invariant under the reflection $\phi : \omega\mapsto (-\omega(-x))_{x\in I}$ and the translations $\psi_b$ for $b\in I$, defined in (\ref{translation}). As a consequence, the measure $P_{N,l}$ is also invariant under these maps, and therefore the matrix valued random variables $\nabla(\hat\omega(\triangle))$ are identically distributed for all $\triangle\in\cT_N$. Thus, for all $\triangle\in\cT_N$, one has

\begin{equation*}
E_{P_{N,l}}[\ ||\ \nabla\hat\omega-l\ \textnormal{Id}\ ||_{L^2(U_N)}^2 \ ] = |\cT_N| \ \lambda(\triangle_{0,1}) E_{P_{N,l}}[\ |\nabla\hat\omega(\triangle)-l\ \textnormal{Id}|^2 \ ].
\end{equation*}

This equation, together with (\ref{maxarea2}), implies

\begin{equation*}
\lim_{l\downarrow1}\sup_{N\in\NNN}\sup_{\triangle\in\cT_N} E_{P_{N,l}}[\ |\nabla\hat\omega(\triangle)-l\ \textnormal{Id}|^2 \ ]=0.
\end{equation*}

By means of the triangle inequality, we see that for all $\triangle\in\cT_N$ and $\omega\in\Omega_{N,l}$

\begin{equation*}
|\nabla\hat\omega(\triangle)- \textnormal{Id}|^2\leq |\nabla\hat\omega(\triangle)-l\ \textnormal{Id}|^2+c_3^2(l-1)^2+2 c_3\ |l-1|\ |\nabla\hat\omega(\triangle)-l\ \textnormal{Id}|
\end{equation*}

with $c_3=|\mathrm{Id}|>0$. For $\omega\in\Omega_{N,l}$, the term $|\nabla\hat\omega(\triangle)-l\ \textnormal{Id}|$ is uniformly bounded for $l\in (1, \epsilon)$ and $N\in\NNN$, which proves the theorem.
\end{proof}

\textbf{Acknowledgement}\quad  I would like to thank Prof. Dr. Merkl for his useful comments and suggestions. Without his support, this work would have not been possible.

\end{document}